\documentclass[letterpaper, 10 pt, conference]{ieeeconf}  
\IEEEoverridecommandlockouts                              
\usepackage{amsmath,amssymb,amsfonts}
\usepackage{cuted}
\usepackage{multicol}
\usepackage{flushend}
\usepackage{xcolor}
\usepackage{dirtytalk}
\usepackage{algorithm2e}
\makeatletter
\renewcommand{\@algocf@capt@plain}{above}
\makeatother
\usepackage{mathtools}
\usepackage{cite}
\usepackage{graphicx}
\usepackage{floatrow}
\mathtoolsset{showonlyrefs}
\usepackage[colorinlistoftodos, textwidth=3cm, shadow]{todonotes}

\PassOptionsToPackage{hyphens}{url}\usepackage{hyperref}
\newtheorem{theorem}{Theorem}[section]
\newtheorem{definition}{Definition}[section]
\newtheorem{assumption}{Assumption}[section]

\newtheorem{lemma}[theorem]{Lemma}

\title{{\LARGE \bf Risk assessment and optimal allocation of security measures under stealthy false data injection attacks}}
\author{
Sribalaji C. Anand$^{1}$, Andr\'e M. H. Teixeira$^{2}$, and Anders Ahl\'en$^{1}$
\thanks{*This work is supported by the Swedish Research Council under the grant 2018-04396 and by the Swedish Foundation for Strategic Research.}
\thanks{$^{1}$ Sribalaji C. Anand and Anders Ahl\'en are with the Department of Electrical Engineering, Uppsala University, PO Box 65, SE-75103, Uppsala, Sweden. {\tt\small sribalaji.anand@angstrom.uu.se}}%
\thanks{$^{2}$ Andr\'e M. H. Teixeira is with the Department of Information Technology, Uppsala University, PO Box 337, SE -75105, Uppsala, Sweden. {\tt\small andre.teixeira@it.uu.se}}%
}
\begin{document}
\maketitle
\thispagestyle{empty}
\pagestyle{empty}
\begin{abstract}
This paper firstly addresses the problem of risk assessment under false data injection attacks on uncertain control systems. We consider an adversary with complete system knowledge, injecting stealthy false data into an uncertain control system. We then use the Value-at-Risk to characterize the risk associated with the attack impact caused by the adversary. The worst-case attack impact is characterized by the recently proposed output-to-output gain. We observe that the risk assessment problem corresponds to an infinite non-convex robust optimization problem. To this end, we use dissipative system theory and the scenario approach to approximate the risk-assessment problem into a convex problem and also provide probabilistic certificates on approximation. Secondly, we consider the problem of security measure allocation. We consider an operator with a constraint on the security budget. Under this constraint, we propose an algorithm to optimally allocate the security measures using the calculated risk such that the resulting Value-at-risk is minimized. Finally, we illustrate the results through a numerical example. The numerical example also illustrates that the security allocation using the Value-at-risk, and the impact on the nominal system may have different outcomes: thereby depicting the benefit of using risk metrics. 
\end{abstract}


\section{Introduction}
Critical infrastructures describe the assets that are vital for the normal operation of society. In general, control systems are an integral part of critical infrastructures. Examples include pH control systems in a bioreactor, frequency control of power generating systems, etc. Due to the vitality of their operation, and partly due to advances in technology, these control systems are monitored regularly through wireless digital communication channels \cite{samad2007system}. And due to the increased use of non-secure communication channels, the control systems are prone to cyber-attacks such as the attack on the Ukrainian power grid, and the Kemuri cyber attack to name a few \cite{hemsley2018history}. Thus there is an increased research interest in the cyber-security of control systems \cite{dibaji2019systems}. 

One of the common recommendations for improving the security of control systems is to follow the risk management cycle: Risk assessment, risk response, and risk monitoring \cite{milovsevic2020security}. The risk assessment step involves careful consideration of risk sources, their likelihood, and their consequences. The consequence can be quantified in terms of impact which can be obtained through simulation or optimization-based methods. The risk response step involves implementing additional measures to minimize the risk if and when necessary. The risk response step can involve either $(i)$ re-designing of the system controller/detectors to be robust against attacks, or $(ii)$ allocating additional security measures such as encrypted communication channels. Finally, the risk monitoring step involves constant monitoring of the risk at acceptable intervals of time. This paper studies the risk assessment and risk response step of the risk management cycle. 

Although the risk assessment and the risk response steps have been studied in the literature \cite[Chapter 2]{milovsevic2020security}, there are some research gaps that are outlined next. Firstly, the majority of the literature considers a deterministic system \cite{milovsevic2019estimating,murguia2020security}. Secondly, the risk frameworks are mostly application-specific. For instance, \cite{CVAR1,5164937} and \cite{liu2016microgrid} determine the risk of cyber-attacks on automatic generation control, power systems, and energy storage systems in smart grids respectively.

To address these limitations, we consider the following setup. A discrete-time (DT) linear time-invariant (LTI) process with uncertainties, an output feedback controller, and an anomaly detector. A stealthy adversary with complete system knowledge injects false data into the sensor or actuator channels for a long but finite amount of time. With this setup, we provide the following contributions.
\begin{enumerate}
    \item We formulate the risk assessment problem using the Value-at-Risk (VaR) as a risk metric and the Output-to-Output Gain (OOG) as an impact metric.
    \item We observe that the risk-assessment problem is NP-hard in general. To this end, we propose an approximate risk assessment problem that is computationally tractable.
    \item We show that the approximate risk assessment problem can be solved by an equivalent convex semi-definite program (SDP). We provide the necessary and sufficient conditions for the (approximate) risk to be bounded.
    \item We provide a preliminary algorithm to optimally allocate security measures using the calculated risk. 
    \item We numerically illustrate that the security measure allocation using the Value-at-risk, and the impact on the nominal system may have different outcomes. We thereby depict the advantage of using risk metrics as suggested in this paper. 
\end{enumerate}

The remainder of the paper is organized as follows. The control system and the adversary are described in Section \ref{description}. The risk assessment problem (RAP) is formulated in Section \ref{PF}. We approximate the RAP and convert it to a convex SDP in Section \ref{Robust o2o}. In Section \ref{sec_allocate} we formulate the security measure allocation problem (SMAP) and provide an algorithm which solves the SMAP for small-scale systems. The results are illustrated through a numerical example in Section \ref{Example}. Finally, we conclude the paper in Section \ref{Conclusion}.

\textit{Notation:} Throughout this paper, $\mathbb{R}, \mathbb{C}, \mathbb{Z}$ and $\mathbb{Z}^{+}$ represent the set of real numbers, complex numbers, integers and non-negative integers respectively. A positive semi-definite matrix $A$ is denoted by $A \succeq 0$. Let $x: \mathbb{Z} \to \mathbb{R}^n$ be a discrete-time signal with $x[k]$ as the value of the signal $x$ at the time step $k$. Let the time horizon be $[0,N]=\{ k \in \mathbb{Z}^+|\; 0 \leq k\leq N \}$. The $\ell_2$-norm of $x$ over the horizon $[0,N]$ is represented as $|| x ||_{\ell_2, [0,N]}^2 \triangleq \sum_{k=0}^{N} x[k]^Tx[k]$. Let the space of square summable signals be defined as $\ell_2 \triangleq \{ x[k]: \mathbb{Z}^+ \to \mathbb{R}^n |\; ||x||^2_{\ell_2, [0,\infty]} < \infty\}$ and the extended signal space be defined as $\ell_{2e} \triangleq \{ x[k]: \mathbb{Z}^+ \to \mathbb{R}^n | \;||x||^2_{[0,N]} < \infty, \forall N \in \mathbb{Z}^+ \}$. For the sake of simplicity, we represent $||x||^2_{\ell_2,[0,\infty]}$ as $||x||^2_{\ell_2}$. For $x \in \mathbb{R}, \left \lceil {x} \right \rceil$ represents the nearest integer $ \geq x$. For any finite set $\mathcal{Q}$, and element of $\mathcal{Q}$ is represented by $q_{(\cdot)}$, and the cardinality of the set is represented by $|\mathcal{Q}|$. 
\section{Problem background}\label{description}
In this section, we describe the control system structure and the goal of the adversary. Consider a closed-loop DT LTI system with a process ($\mathcal{P}$), output feedback controller ($\mathcal{C}$) and an anomaly detector ($\mathcal{D}$) represented by
\begin{align}
    \mathcal{P}: & \left\{
                \begin{array}{ll}
                    x_p[k+1] &= A^{\Delta}x_p[k] + B^{\Delta} \tilde{u}[k]\\
                    y[k] &= C^{\Delta}x_p[k]\\
                    y_p[k] &= C_J^{\Delta}x_p[k] +D_J^{\Delta} \tilde{u}[k]
                \end{array}
                \right. \label{P}\\
    \mathcal{C}: & \left\{
                 \begin{array}{ll}
                     z[k+1] &= A_cz[k]+ B_c\tilde{y}[k]\\
                     u[k] &= C_cz[k] + D_c\tilde{y}[k]
                \end{array}
                \right. \label{C}\\
    \mathcal{D}: & \left\{
                \begin{array}{ll}
                s[k+1] &= A_es[k] +B_e u[k] + K_e\tilde{y}[k]\\
                y_r[k] &= C_es[k] +D_eu[k] + E_e \tilde{y}[k]
                \end{array}
                \right. \label{D}
\end{align}
where  $A^{\Delta} \triangleq A + \Delta A(\delta)$ with $A$ representing the nominal system matrix and $\delta \in \Omega$. Additionally we assume $\Omega$ to be closed,  bounded and to include the zero uncertainty yielding $\Delta A(0) = 0$. The other matrices are similarly expressed. The state of the process is represented by $x_p[k] \in \mathbb{R}^{n_x}$, $z[k] \in \mathbb{R}^{n_z}$ is the state of the controller, $s[k] \in \mathbb{R}^{n_s}$ is the state of the observer, $\tilde{u}[k] \in \mathbb{R}^{n_u}$ is the control signal received by the process, $ u[k] \in \mathbb{R}^{n_u}$ is the control signal generated by the controller, $y[k] \in \mathbb{R}^{n_m}$ is the measurement output produced by the process, $\tilde{y}[k] \in \mathbb{R}^{n_m}$ is the measurement signal received by the controller and the detector, $y_p[k] \in \mathbb{R}^{n_p}$ is the virtual performance output, and $y_r[k] \in \mathbb{R}^{n_r}$ is the residue generated by the detector. The closed-loop system is also shown in Fig. \ref{System}. The reason to adopt uncertainty only in the process is that the parameters of the controller and the detector are chosen by the system operator. However the parameters of the process may not be known to the operator due to a variety of reasons such as, e.g. modelling errors.

In general, the system is considered to have a good performance when the energy of the performance output $||y_p||_{\ell_2}^2$ is small and an anomaly is considered to be detected when the detector output energy $||y_r||_{\ell_2}^2$ is greater than a predefined threshold, say $\epsilon_r$. Without loss of generality (w.l.o.g.), we assume $\epsilon_r = 1$ in the sequel.
\begin{figure}
    \centering
    \includegraphics[width=8.4cm]{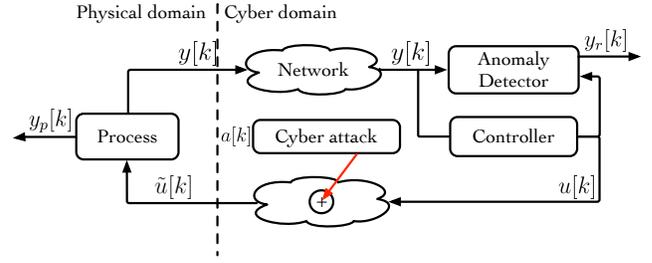}
    \caption{Control system under data injection attack}
    \label{System}
\end{figure}
\subsection{Data injection attack scenario}\label{Attack scenario:sec}
In the closed-loop system described in \eqref{P}-\eqref{D}, we consider that an adversary is injecting false data into the sensors or actuators of the plant but not both. Given this setup, we now discuss the resources the adversary has access to.
\subsubsection{Disruption and disclosure resources}\label{disclosure:sec} 
The adversary can access (eavesdrop) the control and sensor channels and can inject data. This is represented by $\begin{bmatrix}
\tilde{u}^T[k] & \tilde{y}^T[k]
\end{bmatrix}^T=$
\[\begin{bmatrix}
{u}[k]\\
{y}[k]
\end{bmatrix}+\begin{bmatrix}
B_a\\
F_a
\end{bmatrix} a[k], \left[\begin{array}{c|c}
B_a^T & D_a^T
\end{array}\right] \triangleq \left[\begin{array}{c|c}
E_a^T & 0\\
0^T & F_a^T
\end{array}\right]\] 
where $a[k] \in \mathbb{R}^{n_a}$ is the data injected by the adversary. The matrix $E_a (F_a)$ is a diagonal matrix with $E_a(i,i)=1 \;(F_a(i,i)=1)$, if the actuator (sensor) channel $i$ is under attack and zero otherwise. 
\subsubsection{System knowledge} 
We assume that, the system operator knows the bounds of the set $\Omega$ and the nominal system matrices. Next, we assume that the adversary has full system knowledge, i.e., s/he knows the system matrices \eqref{P},\eqref{C}, and \eqref{D}. In reality, it is hard for the adversary to know the system matrices, but this assumption helps to study the worst case.

Defining ${x}[k] \triangleq [ x_p[k]^T \; z[k]^T \; s[k]^T]^T$, the closed-loop system under attack with the performance output and detection output as system outputs becomes
\begin{equation}\label{system:CL:uncertain}
                \begin{array}{ll}
                            {x}[k+1] &= {A}_{cl}^{\Delta}{x}[k] + {B}_{cl}^{\Delta}a[k]\\
                            y_p[k] &= {C}_p^{\Delta}{x}[k] + {D}_p^{\Delta} a[k]\\
                            y_r[k] &= {C}_r^{\Delta}{x}[k] + {D}_r^{\Delta} a[k],\\
                \end{array}
\end{equation}
where the nominal matrices are given by $\left[ \begin{array}{c|c}
A_{cl} & B_{cl}
\end{array}\right] \triangleq $
\[\left[
\begin{array}{c|c}
\begin{matrix}
    A+BD_cC & BC_c & 0\\
    B_cC & A_c & 0\\
    (B_eD_c +K_e)C & B_eC_c & A_e
    \end{matrix} & \begin{matrix}
    BB_a + BD_cD_a \\ B_cD_a \\ (B_eD_c+K_e)D_a 
    \end{matrix}
\end{array}\right]
\]
\begin{align*}
    {C}_p &\triangleq \begin{bmatrix}
    C_J+D_JD_cC & D_JC_c & 0
    \end{bmatrix}\\
    {C}_r &\triangleq \begin{bmatrix}
    (D_eD_c + E_e)C & D_eC_c & C_e
    \end{bmatrix}\\
    {D}_p & \triangleq D_J(D_cD_a+B_a), {D}_r \triangleq (D_eD_c +E_e)D_a.
\end{align*}
In this paper, we consider the adversarial setup where the adversary is omniscient.
\begin{definition}[Omniscient adversary]\label{OA}
An adversary is defined to be omniscient if it knows the matrices in \eqref{system:CL:uncertain}.$\hfill\triangleleft$
\end{definition}

In reality, it is hard for an adversary to know the system matrices of \eqref{system:CL:uncertain} due to the uncertainty. Thus, such an adversarial setup is far from reality but it can help us study a worst-case scenario. For clarity, we assume the following.
\begin{assumption}\label{assume_stable}
The control system \eqref{system:CL:uncertain} is stable $\forall \delta \in \Omega$.$\hfill\triangleleft$
\end{assumption}
\begin{assumption}\label{Ba}
The input matrix has full column rank i.e., $ \nexists \;s \in \mathbb{R}^{n_a} \neq 0$ such that $B_{cl}^\Delta s =0$.$\hfill\triangleleft$
\end{assumption}
\subsubsection{Attack goals and constraints.}
Given the resources the adversary has access to, it aims at disrupting the system's behavior while staying stealthy. The system disruption is evaluated by the increase in energy of the performance output whereas, the adversary is stealthy if the energy of the detection output is below a predefined threshold ($\epsilon_r$). We discuss the attack policy for a deterministic system next.
\subsection{Optimal attack policy for the nominal system}
From the previous discussions, it can be understood that the goal of the adversary is to maximize the performance cost while staying undetected. When the system \eqref{system:CL:uncertain} is deterministic, \cite{teixeira2015strategic} formulates that the attack policy of the adversary as the following non-convex optimization problem
\begin{equation}\label{opti:o2o:uncertain}
\begin{aligned}
||\Sigma||_{\ell_{2e},y_p \leftarrow y_r}^2 \triangleq \sup_{a \in \ell_{2e}} & ||y_p||^2_{\ell_{2}} \\
\textrm{s.t.} & ||y_r||_{\ell_2}^2 \leq 1, x[0]=0,x[\infty]=0,
\end{aligned}
\end{equation}
where $||\Sigma||_{\ell_{2e},y_p \leftarrow y_r}^2$ represents the $OOG$ that characterizes the disruption caused by the attack signal $a$. In \eqref{opti:o2o:uncertain}, the constraint $x[0]=0$ is introduced because the system is assumed to be at equilibrium before the attack. 
\begin{assumption}\label{equil_begin}
The closed-loop system \eqref{system:CL:uncertain} is at equilibrium $x[0]=0$ before the attack commences.$\hfill\triangleleft$
\end{assumption}

We also assume that the adversary has finite amount of energy (similar to $H_{\infty}$ control). Thus, the adversary does not attack the system for an infinite amount of time but stops after a very long time, say $T$. And since the attack stops, the state is brought back to equilibrium. To this end, we introduce the constraint $x[\infty]=0$ in \eqref{opti:o2o:uncertain}.

In the literature, such characterization of the impact of stealthy attacks \eqref{opti:o2o:uncertain} has only been studied for fully known deterministic systems, but not for an uncertain system. Thus, the first goal of the paper is to quantify the impact in terms of risk on the uncertain system \eqref{system:CL:uncertain}. We later describe, in Section \ref{sec_allocate}, as to how the attack impact determined can be used for the benefit of the system operator. 
\section{Problem Formulation}\label{PF}
To quantify the risks of data injection attacks on an uncertain control system, we start by defining a random variable that characterizes the impact as a function of the system uncertainty and the attack vector.

\begin{definition}[Impact random variable]\label{RV}
Let the random variable $X^A(\cdot)$ be defined as
\begin{equation}\label{newimpact}
\begin{aligned}
X^A(a,\delta) \triangleq & \;\Vert y_p(\delta)\Vert_{\ell_2}^2 \times  \mathbb{I}\bigg( \Vert y_r(\delta)\Vert_{\ell_2}^2 \leq 1,x({\delta})[\infty]=0 \bigg)
\end{aligned}
\end{equation}
where $X^A(\cdot)$ is the impact caused on the system \eqref{system:CL:uncertain} with the uncertainty $\delta \in \Omega$ by the attack vector $a \in  \ell_{2e}$, $\mathbb{I}$ is the indicator function,  $y_p(\delta), y_r(\delta)$ and $x({\delta})$ are the performance, residue output and state of the system with the isolated uncertainty $\delta$. Here, the signals $y_p(\delta), y_r(\delta)$ and $x(\delta)$ are also functions of the attack vector $a$.$\hfill\triangleleft$
\end{definition}

With the random variable defined in \textit{Definition \ref{RV}}, we next formulate the risk assessment problem. Consider the data injection attack scenario where the parametric uncertainty $\delta \in \Omega$ of the system is known to the adversary but not to the system operator. The system operator has knowledge only about the bounds of the set $\Omega$. Recall that such a scenario is far from reality, but such a setup helps us study the worst case. Under this setup, the adversary can cause high disruption by remaining stealthy because the adversary will be able to inject attacks by solving the optimization problem
\begin{equation}\label{eqq4}
\begin{aligned}
 ||\tilde{\Sigma}(\delta)||_{\ell_{2e},y_p \leftarrow y_r}^2 \triangleq  \sup_{a_{\delta} \in \ell_{2e}}\; & X^A(a_{\delta},\delta),
\end{aligned}
\end{equation}
where $a_{\delta}$ represents the attack vector corresponding to the uncertainty $\delta$. Since the system operator does not know the uncertainty $\delta$, $||\tilde{\Sigma}(\delta)||_{\ell_{2e},y_p \leftarrow y_r}^2$ can be interpreted as a random variable. Thus, for the system operator, the best option is to assess the risk associated with the impact based on a risk metric. In this paper, we adopt the risk metric VaR \cite{mausser1999beyond}.
\begin{definition}[Value-at-Risk (VaR)]\label{def_Var}
Given a random variable $X$ and $\beta \in (0,1 )$, the VaR is defined as
\[
\text{VaR}_{\beta}(X) \triangleq \inf \{x|\mathbb{P}[X\leq x] \geq 1-\beta\}.
\]
With a specified probability level $\beta \in (0,1)$, $\text{VaR}_{\beta}$ is the lowest amount of $x$ such that with probability $1-\beta$, the random variable, $X$, does not exceed $x$. $\hfill\triangleleft$
\end{definition}

Therefore, by calculating $\text{VaR}_{\beta}$, one can ensure that the probability that the value (X) exceeds $\text{VaR}_{\beta}$ is less than or equal to $\beta$. In our setting, the system operator is interested in determining the $\text{VaR}_{\beta}$ given a small $\beta$ such that the impact rarely exceeds $\text{VaR}_{\beta}$. Given that the impact caused by the adversary on \eqref{system:CL:uncertain} is characterized by $||\tilde{\Sigma}(\delta)||_{\ell_{2e},y_p \leftarrow y_r}^2 $, the $\text{VaR}_{\beta}( \cdot )$ can be obtained, using \textit{Definition \ref{def_Var}}, by solving
\begin{equation}\label{uncertain:system}
\begin{aligned}
\gamma_{OA} \triangleq \inf_{\gamma}\; & \gamma\\
 \textrm{s.t.} \; & \mathbb{P}_{\Omega} ( ||\tilde{\Sigma}(\delta)||_{\ell_{2e},y_p \leftarrow y_r}^2 \leq \gamma ) \geq 1-\beta,
\end{aligned}
\end{equation}
where $\gamma_{OA}$ represents the VaR associated with the impact caused by an \textbf{O}mniscient \textbf{A}dversary.

Although VaR is not extensively used in the literature \cite{mausser1999beyond}, it is used here to only assess the worst-case risk even though this setup may be deemed unrealistic. Since $\mathbb{P}$ in \eqref{uncertain:system} operates over the continuous space $\Omega$, the optimization problem is computationally intensive or in general NP-hard \cite[Section 3]{ben1998robust}. Besides \eqref{eqq4} is a non-convex \cite{teixeira2015strategic}. In Section \ref{Robust o2o} we discuss a method to solve \eqref{uncertain:system} approximately and efficiently.
\section{Risk assessment}\label{Robust o2o}
To recall, \eqref{uncertain:system} is computationally intensive since $\Omega$ is a continuum. To this end, in this section we determine an approximate solution to \eqref{uncertain:system} and also provide some probabilistic certificates using the scenario approach introduced in \cite{calafiore2007probabilistic}.
\subsection{Discrete uncertainty set}
In this section, we consider a discrete uncertainty set $\Omega$. The results of this section is the basis for addressing a continuous uncertainty set next. For brevity, given a sampled uncertainty $\delta_i \in \Omega$, we define $\tilde{\Sigma}_{p,i} \triangleq ({A}_{cl,i}, {B}_{cl,i}, {C}_{p,i}, {D}_{p,i})$ and  $\tilde{\Sigma}_{r,i} \triangleq ({A}_{cl,i}, {B}_{cl,i}, {C}_{r,i}, {D}_{r,i})$ with $y_p(\delta_i)=y_{pi},y_r(\delta_i)=y_{ri}$ and $x(\delta_i)=x_{i}$ as the outputs and states of $\tilde{\Sigma}_{p,i}$ and  $\tilde{\Sigma}_{r,i}$ correspondingly. For such an isolated uncertainty, \eqref{eqq4} can be rewritten as
\begin{equation}\label{pp1}
\begin{aligned}
 ||\tilde{\Sigma}(\delta_i)||_{\ell_{2e},y_p \leftarrow y_r}^2 &\triangleq  \sup_{a_i \in \ell_{2e}}\; X^A(a_i,\delta_i)\\
 &= \left\{\begin{aligned}
  \sup_{a_i \in \ell_{2e}} & \Vert y_{p,i}\Vert_{\ell_2}^2\\
\textrm{s.t.} \; & \Vert y_{r,i}\Vert_{\ell_2}^2 \leq 1,x_i[\infty]=0
\end{aligned}\right\}
\end{aligned}
\end{equation}
where $a_i$ is the attack vector corresponding to the uncertainty $\delta_i$. The optimization problem \eqref{pp1} has two disadvantages: it is non-convex and intractable (since the optimizer is infinite-dimensional). To this end, we can use the Lagrange dual function to reformulate the non-convex problem into its dual-counterpart. Furthermore, we can use dissipative system theory to convert the intractable non-convex problem to a convex problem with LMI constraints which is tractable. This reformulation is presented in \textit{Lemma \ref{lemm1}}.
\begin{lemma}\label{lemm1}
The optimization problem \eqref{pp1} is equivalent to the convex SDP \eqref{s2}.
\begin{equation}\label{s2}
\min_{\gamma_i \geq 0 ,P_i = P_i^T} \; \left\{ \gamma_i \;  \Big| M (\gamma_i,\delta^i,P_i) \preceq 0 \right\}
\end{equation}
where $M (\gamma_i,\delta^i,P_i) \triangleq 
\begin{bmatrix} A_{cl,i}^TP_iA_{cl,i}-P_i & A_{cl,i}^TP_iB_{cl,i} \\ B_{cl,i}^TP_iA_{cl,i} & B_{cl,i}^TP_iB_{cl,i} \end{bmatrix} 
+ \begin{bmatrix}
C_{p,i}^T\\
D_{p,i}^T
\end{bmatrix}
\begin{bmatrix}
C_{p,i} & D_{p,i}
\end{bmatrix} - \gamma_i
\begin{bmatrix}
C_{r,i}^T\\
D_{r,i}^T
\end{bmatrix}
\begin{bmatrix}
C_{r,i} & D_{r,i}
\end{bmatrix}.$\\
\end{lemma}
\begin{proof}
The result is similar to \cite[Theorem 1]{teixeira2015strategic} and thus the proof is omitted.
\end{proof}
Next, we discuss the the conditions for boundedness of $||\tilde{\Sigma}(\delta_i)||_{\ell_{2e},y_p \leftarrow y_r}^2$ in the \textit{Lemma \ref{bound_OA_single}}.
\begin{lemma}[Boundedness]\label{bound_OA_single}
Consider the closed-loop system \eqref{system:CL:uncertain} with an uncertainty $\delta_i$ which is known to the adversary. Then, the optimal value of \eqref{s2} is bounded if and only if one of the following conditions hold:
\begin{enumerate}
    \item The system $\tilde{\Sigma}_{r,i}$ has no zeros on the unit circle.
    \item The zeros on the unit circle of the system $\tilde{\Sigma}_{r,i}$ (including multiplicity and input direction) are also zeros of $\tilde{\Sigma}_{p,i}$.
\end{enumerate}
\end{lemma}
\begin{proof}
See appendix.
\end{proof}

\textit{Lemma \ref{bound_OA_single}} states that the attack impact caused by an omniscient adversary on \eqref{system:CL:uncertain} with an uncertainty $\delta_i$ is bounded if either, there does not exist an attack vector which makes the output $y_r$ identically zero, or all attack vectors which yields $y_r$ identically $0$ also yields $y_p$ identically zero. 

In this section, we formulated the results on characterizing the attack impact by a convex SDP and the condition for its boundedness for an isolated discrete uncertainty. Next, the approach discussed in this section for risk assessment is extended when considering a continuum of uncertainties. 
\subsection{Continuous uncertainty set}
The optimization problem \eqref{s2} can be directly extended to solve \eqref{uncertain:system} only when we consider a discrete set $\Omega$. But \eqref{uncertain:system} that we intend to solve operates over a continuous set $\Omega$. Using the framework of scenario-based reliability estimation \cite{calafiore2007probabilistic}, $\Omega$ can be approximated with a finite set. With this scenario-based framework we revisit \eqref{uncertain:system} in \textit{Theorem \ref{Thm1}}. 
\begin{theorem}\label{Thm1}
Let $\epsilon_1 \in (0,1)$ represent the accuracy with which the probability operator $\mathbb{P}_{\Omega}$ in \eqref{uncertain:system} is to be approximated. Let $\beta_1 \in (0,1)$ represent the confidence with which the accuracy $\epsilon_1$ is guaranteed, i.e.,
\[ \mathbb{P}\{ \vert \mathbb{P}_{\Omega}(||\tilde{\Sigma}(\delta)||_{\ell_{2e},y_p \leftarrow y_r}^2 \leq \gamma ) - \hat{\mathbb{P}}_{N_1} \vert \geq\epsilon_1\} \leq \beta_1.\]
Here $ \hat{\mathbb{P}}_{N_1}$ represents the approximation of the probability operator $\mathbb{P}_{\Omega}$ in \eqref{uncertain:system} defined as
\begin{equation}
    \hat{\mathbb{P}}_{N_1} \triangleq \frac{1}{N_1}\sum_{i=1}^{N_1}\mathbb{I}\left( ||\tilde{\Sigma}(\delta_i)||_{\ell_{2e},y_p \leftarrow y_r}^2 \leq \gamma \right), \hfill \text{where}
\end{equation}
\begin{equation}\label{sam1}
N_1 \geq \frac{1}{2\epsilon_1^2}\text{log}\frac{2}{\beta_1}.
\end{equation}
Then, the $\text{VaR}_{\beta}$ defined in \eqref{uncertain:system} can be obtained with an accuracy $\epsilon_1$ and confidence $\beta_1$ by solving
\begin{equation}\label{problem:decoupled}
\hat{\gamma}_{OA} = \left\{\begin{aligned}
\min & \;\;\gamma\\
\textrm{s.t.} \quad & \frac{1}{N_1} \sum_{i=1}^{N_1} \mathbb{I}\left( \gamma_i \leq \gamma \right) \geq 1-\beta,
\end{aligned}\right\}
\end{equation}
where $\hat{\gamma}_{OA}$ represents the $\text{VaR}_{\beta}$  with an accuracy $\epsilon_1$, and $\mathbb{I}$ is the indicator function. The value of  $\gamma_i,i=1,\dots, N_1$ is obtained by solving the convex SDP \eqref{s2}.
\end{theorem}
\begin{proof}
See appendix.
\end{proof}

\textit{Theorem \ref{Thm1}} states that, to solve \eqref{problem:decoupled}, one could solve the optimization problem \eqref{s2} for $N_1$ unique realizations of the uncertainty. Since strong duality holds between \eqref{pp1} and \eqref{s2}, the optimal value of the dual optimization problem \eqref{problem:decoupled} indeed provides the VaR$_{\beta}$ with an accuracy $\epsilon_1$ and confidence $\beta_1$. Thus \textit{Theorem \ref{Thm1}} provides a method to determine the risk through a sampled uncertainty set, and provides apriori probabilistic certificates on the accuracy and confidence of the operator $\mathbb{P}$. Next, by extending \textit{Lemma \ref{bound_OA_single}}, the condition for boundedness of \eqref{problem:decoupled} is stated in \textit{Lemma \ref{bound_OA_multi}}. 
\begin{lemma}[Boundedness]\label{bound_OA_multi}
Consider $N_1$ independent and identically distributed realizations of \eqref{system:CL:uncertain}, each with an uncertainty $\delta_i$. The optimal value of \eqref{problem:decoupled} with these $N_1$ system realizations is bounded iff the optimal value of \eqref{s2} is bounded for at least $\left \lceil{N_1(1-\beta)}\right \rceil $ system realizations. 
\end{lemma}
\begin{proof}
See appendix.
\end{proof}

The interpretation of \textit{Lemma \ref{bound_OA_multi}} is that the system operator tolerates a fraction ($\beta \times 100 \%$) of cases where the impact \eqref{s2} is unbounded. Conversely, even-though the impact is unbounded for certain realization of uncertainty, the risk will still be bounded. This allows the system operator to be less pessimistic: In the sense that, even though the attack impact in certain scenarios can be very high, the risk evaluated by the operator will be bounded. In the next section, we  briefly discuss the benefit of determining the risk and how it can be used by the system operator. 
\section{Optimal allocation of security measures}\label{sec_allocate}
In this section, we discuss how the RAP discussed in Section \ref{Robust o2o} can be used for the benefit of the operator. That is, after determining the risk, the operator might be interested in the question \say{If the risk value is not acceptable what actions steps can be taken?}. To this end, the determined risk can be used in two ways. On one hand, the operator can use the risk as a metric to design the controllers/detectors of the system optimally \cite{8618886}. On the other hand, the risk can be used to allocate the security measures optimally \cite[Chapter 5]{milovsevic2020security} which is the problem considered here. 

Let $n_w$ be the number of secure resources. In reality, secure resources refer to some form of secure communication channels for the sensors and actuators such that an attack cannot occur. If the number of secure resources ($n_w$) is equal to the number of actuators and sensors ($n_u+n_y$), then the SMAP is solved trivially. However, in general $n_w << n_u+n_y$ since secure communication channels are expensive\footnote{by expensive we here mean encryption and processing costs}. Thus, we discuss a method to optimally allocate the security measures when they are limitedly available. 

To formulate the problem, we define the following. The set of all sensors (actuators) is represented by $\mathcal{S} (\mathcal{A})$, where $|\mathcal{S}|=n_y (|\mathcal{A}|=n_u)$. The set of all vulnerabilities is represented by $\mathcal{V}$. Any sensor (actuator) is a vulnerability if the operator believes that an adversary might be able to access the sensor (actuator) channels. Thus $|\mathcal{V}| = n_v \leq n_y+n_u$. Let the set of secure resources be represented by $\mathcal{W}$ where $|\mathcal{W}| = n_w$. And as discussed before $n_w << n_v \leq n_y+n_u$. Then the SMAP has the following structure. 

Firstly, the operator aims at optimizing the risk metric. Secondly, if an actuator (sensor) $i \in \mathcal{V}$ is secure, then the corresponding actuator (sensor) channels cannot be accessed by the adversary (C$1$). Recall from Section \ref{description}, that the matrix $E_a (F_a)$ is a diagonal matrix with $E_a(i,i)=1 \;(F_a(i,i)=1)$, if the actuator (sensor) channel $i$ is under attack and zero otherwise. And, as discussed before, the operator can only secure $n_w$ actuators (sensors) at most (C$2$). To this end, the optimal SMAP under actuator attacks can be formulated as
\begin{equation}\label{allocate_P1}
    \left\{\hat{\gamma}_{OA}^*, W^* \right\}= \left\{\begin{aligned}
    \inf_{z_i} \;& \hat{\gamma}_{OA}(z)\\
    \text{s.t.}\; & (\text{C}1) E_a(i,i) = z_i,\;\forall i \in \mathcal{V}\\
    & (\text{C}2) \sum_{i \in \mathcal{V}}(1-E_a(i,i)) \leq n_w\\
    & z_i \in \{0,1\}, \;\forall i
    \end{aligned}\right\}
\end{equation}
where, $\hat{\gamma}_{OA}^*$ is the optimal risk after the security measures are allocated, the corresponding optimal vulnerabilities to be protected are represented by $W^*$, and where the constraint C$2$ considers that a vulnerable actuator is protected if and only if the corresponding actuator has $E_a(i,i)=0$. Similarly, when the sensors are under attack, the optimal SMAP can be formulated by replacing $E_a(\cdot)$ by $F_a(\cdot)$ in \eqref{allocate_P1}.

The optimization problem \eqref{allocate_P1} is hard to solve since it is a combinatorial problem. That is, the operator has to search through the whole set of $\mathcal{V}$ to secure $n_w$ vulnerabilities. And it is well known that combinatorial problems with a large search space ($\mathcal{V}$) are NP-hard in general. Thus, providing a heuristic to solve \eqref{allocate_P1} is beyond the scope of this paper and is left for future work. Interested readered are referred to \cite[Chapter 5]{milovsevic2020security}. However, we provide a scheme to solve \eqref{allocate_P1} when $|\mathcal{V}|$ is small in \textbf{Algorithm \ref{algo1}}.
\begin{algorithm}
\SetAlgoLined
\caption{Algorithm to solve SMAP}
\textbf{Initialization}: $\beta,\epsilon, \Omega, \epsilon,\mathcal{V},n_w$ an empty list $\gamma_{l}$\\
\textbf{Step 1:} Determine $\mathcal{Q}$ which is the set of all subsets of $\mathcal{V}$ with maximum cardinality $n_w$.\\
\textbf{Step 2:} For all $q_{(\cdot)} = \mathcal{Q}$, do:\\
\begin{enumerate}
        \item Set $E_a(i,i) = 0$ if $i \in q_{(\cdot)}$ and $1$ otherwise.
        \item Determine $\hat{\gamma}_{OA}$ with the new $E_a$.
        \item Append $\gamma_{l}$ with $\hat{\gamma}_{OA}$
\end{enumerate}
\textbf{Step 3:} Determine the minimum of $\gamma_{l}$ which is $\gamma_{OA}^*$.\\
\textbf{Step 4:} Determine the corresponding $q_{(\cdot)}^*$ which yields $\gamma_{OA}^*$.\\
\textbf{Step 5:} Set $W^* \triangleq q_{(\cdot)}^*$.\\
\KwResult{$\hat{\gamma}_{OA}^*,W^*$}
\label{algo1}
\end{algorithm}

The algorithm first determines all possible subsets of the vulnerabilities with the maximum cardinality of $n_w$. Then, the operator determines the maximum attack impact caused by the adversary when these various subsets of vulnerabilities are protected. It then selects the set of vulnerabilities ($W^*$) which yields the minimum attack impact ($\gamma_{OA}^*$).

In this section we discussed how the risk determined in Section \ref{Robust o2o} can be used by the system operator to optimally allocate the security measures. In the next section, we will illustrate what results can be obtained by a simple example.
\section{Numerical example}\label{Example}
In this section, the effectiveness of the proposed \textbf{Algorithm \ref{algo1}} is illustrated through a numerical example. Consider the system described in \eqref{system:CL:uncertain} with $C = C_J^T =I_3, E_a= I_2$
\begin{align}
\left[
\begin{array}{c|c}
A & B^{\Delta}
\end{array}
\right] = 
\left[ 
\begin{array}{c|c}
\begin{matrix}
1 & 0 & 1\\
1 & 0.5 & 0\\
0 & 1 & -0.5
\end{matrix}     & \begin{matrix}
1.5+\delta & 0 \\ 0.3 & 0 \\0 & 1
\end{matrix}
\end{array}
\right],
\end{align} 
$\Omega \triangleq [-0.5,\;0.5], A_e=A, B_e=B, C_e=C,\left[
\begin{array}{c|c}
D_c^T & K_e
\end{array}
\right]=$
\begin{align}
\left[ 
\begin{array}{c|c}
\begin{matrix}
   -0.066  &  0.178\\ 
   0.047 & 0.940 \\
    0.524 &  -1.346
\end{matrix}     & \begin{matrix}
    0.393&0&1\\
    1 &-0.048 & 0\\
    0 &1& -0.996
\end{matrix}
\end{array}
\right],
\end{align} 
and all the other unspecified matrices are zero. In the system description, only the matrix $B^{\Delta}$ is a function of the uncertain variable. And the system has no uncertain zeros on the unit circle, which makes the condition of \textit{Lemma \ref{lemm1}} hold $\forall \delta \in \Omega$. Thus, for a nominal system with $\delta=0$, the OOG obtained by solving \eqref{opti:o2o:uncertain} is $||\Sigma||_{\ell_{2e},y_p \leftarrow y_r}^2 = 197.76$.

Let $\epsilon_1 = 0.05, \beta_1 = \beta = 0.1$ and $N_1=235$ which satisfies \eqref{sam1}. Here $1-\epsilon_1$ represents the accuracy of the approximation of the probability operator in determining the VaR$_{\beta}$. And $\beta_1$ represents the guarantee. We then solve \eqref{problem:decoupled} and obtain $\gamma_{OA} = 347.15$. 

The optimization problem \eqref{problem:decoupled} is solved as follows. The set $\Omega$ is sampled for $N_1$ samples. The value of $\gamma_i \;\forall i=\{1,\dots,N_1\}$ is obtained by solving the SDP \eqref{s2}. From these values of $\gamma_i$, we choose $\gamma_{OA}$ such that the $\mathbb{P} (\gamma_i \geq \gamma_{OA}) = \beta$. Maintaining $\epsilon_1$ and $\beta_1$ constant, for varying values of $\beta$, the value of $\gamma_{OA}=\text{VaR}_{\beta}(X)$ is shown in Fig. \ref{fig2}.

\begin{figure}
    \centering
    \includegraphics[width=8.4cm]{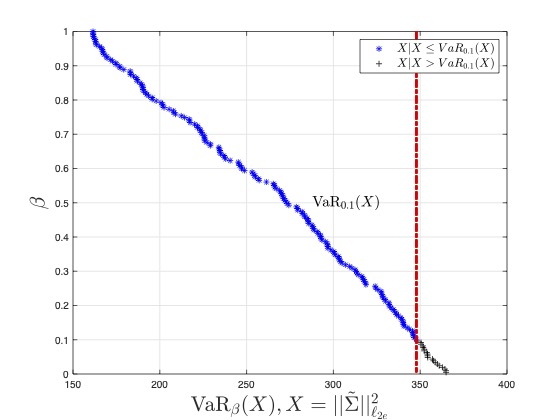}
    \caption{The parameter $\beta$ is shown on the Y-axis and the corresponding VaR$_{\beta}$ on the X-axis. The red line indicates VaR$_{0.1}$. The blue dots denotes the value of the impact random variable $X|X < VaR_{0.1}$, whereas the black dots denotes the value of the impact random variable $X|X > VaR_{0.1}$. It can be seen that the probability that $X > VaR_{0.1}$ is low when $\beta$ is small.}
    \label{fig2}
\end{figure}

Fig. \ref{fig2} depicts that the value of the impact is greater than the VaR with a probability $\beta$ and confidence $1-\epsilon_1$. To recall, by determining VaR$_{\beta}$ for a given $\beta$, the operator can ensure that the impact of any stealthy attack impact is greater than VaR$_{\beta}$ with probability $\beta$. After determining the VaR, the operator decides if the risk is acceptable or not. If the risk is not acceptable, the risk can be minimized for a given $\beta$ by implementing additional security measures.  

We next use the risk metric determined to allocate the security measures. Let $\mathcal{V} = \mathcal{S} = \{S1,S2,S3\}$. That is, we assume that all sensors communication channels are vulnerable to attacks. Then, we determine the risk when there are no security measures available ($|\mathcal{X}| = 0$). Next, we determine the risks when there is only one security measure available ($|\mathcal{X}| = 1$). That is, we determine the risks corresponding to the setup where either $S1,S2$ or $S3$ is protected. And finally, we determine the risks when there are two security measure available ($|\mathcal{X}| = 2$). That is, we determine the risk corresponding to the setup where either $\{S1, S2\},\{S2, S3\}$ or $\{S3,S1\}$ are protected. The risks are depicted in the left diagram of Fig. \ref{fig3} in blue, where the text on the top of the bar depicts the sensors that are protected. From Fig \ref{fig3}, we can conclude that $(i)$ when $|\mathcal{X}| = 1$, it is optimal to secure $S3$ since it minimizes the risk the most, and $(ii)$ when $|\mathcal{X}| = 2$, it is optimal to secure $S2,$ and $S3$.

We repeat the procedure for the actuators where $\mathcal{V} = \mathcal{A} = \{A1,A2\}$. That is, we assume that all actuator communication channels are vulnerable to attacks. The risks are depicted in the right diagram of Fig. \ref{fig3} in blue. From Fig. \ref{fig3}, we can conclude that $(i)$ when $|\mathcal{X}| = 1$, it is optimal to secure $A1$ (actuator 1), and $(ii)$ it is much more riskier to leave the sensors unprotected since the risk of unprotected sensors is higher than unprotected actuators. 

Finally, we show the advantage of using the risk metric. We use the impact on the nominal system as a metric to allocate the protection resources. That is, instead of solving the SMAP with the risk determined from \eqref{problem:decoupled}, we simply solve the optimization problem \eqref{s2} for the nominal system and use it as a metric to allocate the security measure. To this end, we determine the impact on the nominal system when  $|\mathcal{X}| = 0, |\mathcal{X}| = 1$, and $|\mathcal{X}| = 2$. The corresponding impact are shown in Fig. \ref{fig3} in red. It can be seen that the conclusion that we drew using the risk metric are not reproducible when we use the nominal impact as a metric. For instance, when the sensors are under attack and $|\mathcal{X}| = 1$, the conclusion from the risk metric is to protect $S3$, whereas if we use the nominal impact, we end up protecting $S2$. Similarly, when the actuators are under attack and $|\mathcal{X}| = 1$, the conclusion from the risk metric is to protect $A1$, whereas if we use the nominal impact, we end up protecting $A2$. 

\begin{figure}
    \centering
    \includegraphics[width=8.4cm]{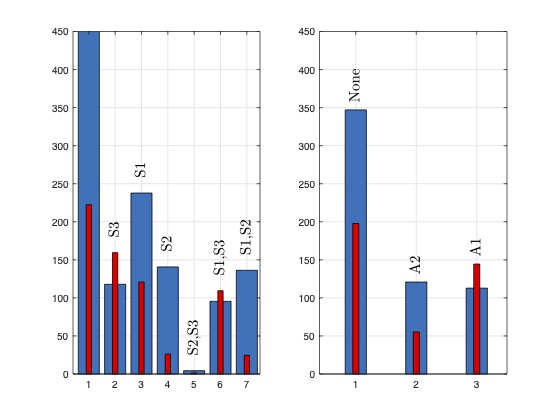}
    \caption{The VaR$_{0.1}$ after protecting various combination of sensors (actuators) are depicted on the left (right) figure in blue. The text on the top of each bar denotes the sensor (actuator) that is protected. For instance, ``None" represents that none of the sensors are protected. The bar at position ``1" of the figure in left corresponds to the risk when none of the sensors are protected and the corresponding risk was found to be 9081.4. The Y axis of the figure is not extended to show this value since it would affect clarity of the figure. The plots in red represent the impact on the nominal system after the security measures are allocated using the impact on the nominal system as a metric.
    }
    \label{fig3}
\end{figure}

\section{Conclusion}\label{Conclusion}
In this paper, we first addressed the RAP of false data attacks injected by an omniscient adversary on uncertain control systems. We formulated the RAP and observed that it is a non-convex infinite robust optimization problem. Using the theory of dissipative systems and scenario approach, we approximated the RAP as a convex SDP with probabilistic certificates. The necessary and sufficient conditions for the risk to be bounded were also formulated. Secondly, we consider the optimal SMAP. We used the risk determined as a metric to formulate the SMAP. We provide a preliminary algorithm to solve the allocation problem. The results were depicted through a numerical example. Future works include $(i)$ considering a process with process and measurement noise, and $(ii)$ providing a more detailed solution for the allocation problem. 

\appendix
\renewcommand{\thetheorem}{A.\arabic{theorem}}
\setcounter{theorem}{0}
\section*{Proof of \textit{Lemma \ref{bound_OA_single}}}    
\begin{proof}\label{bound_1}
To recall, the optimization problem \eqref{s2} was formulated using $3)$ in \cite[Proposition 2]{Truncated} where $y_1 = \sqrt{\gamma} y_r$ and $y_2 = y_p$. Due to the equivalency between $3)$ and $4)$ of \cite[Proposition 2]{Truncated}, the Frequency Domain Inequality (FDI) $G_1(\Bar{z})^TG_1(z)-G_2(\Bar{z})^TG_2(z) \succeq 0$ should hold $\;\forall \;|z|=1$. Since we know that $y_1 = \sqrt{\gamma} y_r$ and $y_2 = y_p$, we can deduce that $G_1(z)$ corresponds to $\sqrt{\gamma}\tilde{G}_{r,i}(z)$ and $G_2(z)$ to $\tilde{G}_{p,i}(z)$ in FDI, where $\tilde{G}_{r,i}(z_1) \triangleq C_{r,i}(z_1I-A_{cl,i})^{-1}B_{cl,i}+D_{r,i}$ and $\tilde{G}_{p,i}(z_1) \triangleq C_{p,i}(z_1I-A_{cl,i})^{-1}B_{cl,i}+D_{p,i}$. Thus, \eqref{s2} can be rewritten as
\begin{equation} \label{fi2}
\inf_{\gamma_i \geq 0}\left\{\gamma_i \Bigg| \underbrace{\gamma_i\tilde{G}_{r,i}(\bar{z})^T\tilde{G}_{r,i}({z})-\tilde{G}_{p,i}^T(\bar{z})\tilde{G}_{p,i}({z})}_{H(z,\gamma_i)}\succeq 0,\forall|z|=1\right\}
\end{equation}
which is equivalent to
\begin{equation}\label{lem_s1}
\inf_{\gamma_i\geq0}\left\{\gamma_i \Bigg| x^HH(z,\gamma_i)x \geq 0,\forall|z|=1\right\}
\end{equation}
Next, let us define the following sets
\begin{align}
    \mathcal{Z}_{pr} & \triangleq \{x \in \mathbb{C}^{n_a} \;|\; \tilde{G}_{r,i}({z})x = 0 , \tilde{G}_{p,i}({z})x = 0\},\\
    \mathcal{Z} &\triangleq \{x \in \mathbb{C}^{n_a} \;|\; \tilde{G}_{r,i}({z})x \neq 0 , \tilde{G}_{p,i}({z})x \neq 0\},\\
    \mathcal{Z}_r & \triangleq \{x \in \mathbb{C}^{n_a} \;|\; \tilde{G}_{r,i}({z})x = 0 , \tilde{G}_{p,i}({z})x \neq 0\},\\
    \mathcal{Z}_p & \triangleq \{x \in \mathbb{C}^{n_a} \;|\; \tilde{G}_{r,i}({z})x \neq 0 , \tilde{G}_{p,i}({z})x =0\}.
\end{align}
In the above definitions, each set corresponds to a combination of two logical conditions of $\tilde{G}_{r,i}({z})x$ and $\tilde{G}_{p,i}({z})x$. Therefore, the union of all four sets explores all possible combinations of the two logical conditions, and thus their union corresponds to the entire set $\mathbb{C}^{n_a}$. 

For any given $z$ such that $|z|=1$ , if $ x \in \mathcal{Z}_p$, choosing $\gamma =0$ satisfies the constraint of \eqref{lem_s1}. Similarly, if $ x \in \mathcal{Z}$, then $\gamma = \sup_{|z|=1, x \in \mathcal{Z}} \frac{x^H\big\{\tilde{G}_{r,i}^T(\bar{z})\tilde{G}_{r,i}({z})\big\}x}{x^H\big\{\tilde{G}_{p,i}^T(\bar{z})\tilde{G}_{p,i}({z})\big\}x}$. This ratio is well defined since the denominator cannot become zero (since $x \in \mathcal{Z}$), and the ratio is bounded since we have assumed that the transfer functions $\tilde{G}_{r,i}({z})$ and $\tilde{G}_{p,i}({z})$ are always stable (\textit{Assumption \ref{assume_stable}}). Therefore, we have proven that the value of \eqref{lem_s1} is bounded whenever $x \in \mathcal{Z}_{p} \cup \mathcal{Z}$. We next begin by proving that the lemma statements are sufficient for \eqref{lem_s1} to be bounded whenever $x \in \mathcal{Z}_{r} \cup \mathcal{Z}_{pr}$.

\textit{Sufficiency:} 
Assume that condition $(1)$ of the lemma statement holds. By definition of a zero $ \forall |z|=1, \nexists s \neq 0 \in \mathbb{C}^{n_a}\;\text{s.t.}\;\tilde{G}_{r,i}({z})s = 0$. Thus it follows that $\mathcal{Z}_r = \mathcal{Z}_{pr} = \emptyset$.  

Assume that condition $(2)$ of the lemma statement holds. Then by definition of a zero $ \forall |z|=1, \nexists s \neq 0$ such that $\tilde{G}_{r,i}({z})s = 0$ and $\tilde{G}_{p,i}({z})s \neq 0$. Thus it follows that $\mathcal{Z}_r = \emptyset$. And if $x \in \mathcal{Z}_{pr}$, then picking $\gamma =0$ simply satisfies the constraint of \eqref{lem_s1}. This concludes the proof on sufficiency.\\ 
\textit{Necessity:} We prove by contradiction. Assume that there exists a bounded $\gamma$ which solves the optimization problem \eqref{lem_s1}. And we also assume that there exists a complex number $z_1$ on the unit circle which is a zero of the system $\tilde{\Sigma}_{r,i}$ (including multiplicity and input direction) but are not zeros of $\tilde{\Sigma}_{p,i}$. By definition of a zero, it holds that $\exists s \neq 0$ such that $\tilde{G}_{r,i}(z_1)s=0,\;\tilde{G}_{p,i}(z_1)s \neq 0$. Thus, $\mathcal{Z}_{rp} \neq \emptyset$ and becomes a part of the feasible set for $x$. Then, if $z=z_1$ and $x = s$, the constraint of \eqref{lem_s1} can be rewritten as $-s^H\tilde{G}_{p,i}^T(\bar{z}_1)\tilde{G}_{p,i}({z_1})s \geq 0$ which cannot hold since $\tilde{G}_{p,i}({z_1})s \neq 0$. That is, the feasibility set of \eqref{lem_s1} is empty which contradicts our assumption. In terms of the primal problem \eqref{s2}, it means that there cannot exist a bound to its optimal value. This concludes the proof.
\end{proof}
\section*{Proof of \textit{Theorem \ref{Thm1}}} \label{App1}     
Before presenting the proof, an introduction to scenario-based performance level estimation \cite{calafiore2007probabilistic} is provided.
\subsection*{Scenario-based performance level estimation}\label{thm1_sample}
Consider a function $f(\delta),\; \delta \in \Delta$. Let $\epsilon \in (0,1)$, then the performance level estimation problem aims to estimate a performance level $\gamma^*$ with a reliability $1-\epsilon$ such that 
\begin{equation}\label{appro-reli}
\gamma^* \triangleq  \inf \left\{\gamma\; \Big|\; \mathbb{P}_{\Delta}\{f(\delta) \leq \gamma\} \geq 1-\epsilon.\right\}
\end{equation}
The problem \eqref{appro-reli} is computationally intensive. The scenario-based probability estimation algorithm (PEA), provides a randomized approach to approximate the probability operator $\mathbb{P}$. This algorithm is described below.
\begin{enumerate}
    \item Choose a confidence level $\beta_1 \in (0,1)$, and an accuracy level $\epsilon_1 \in (0,1)$.
    \item Choose $N_1$ such that \eqref{sam1} holds.
 \item Extract $N_1$ i.i.d. samples $\delta^1,\dots,\delta^N$ from $\Delta$ and evaluate $f(\delta^1),\dots,f(\delta^{N_1})$.
\item Then it holds that \[ \mathbb{P}\{ \vert \mathbb{P}_{\Delta}\{f(\delta) \leq \gamma\} - \hat{\mathbb{P}}_N \vert \geq\epsilon_1\} \leq \beta_1,\]
where $\hat{\mathbb{P}}_N \triangleq \frac{1}{N_1} \sum_{i=1}^{N_1} \mathbb{I}\left( f(\delta^i) \leq \gamma \right).$
Thus, in this step, an equivalent probability operator $ \hat{\mathbb{P}}_N$ is formulated in place of $\mathbb{P}_{\Delta}$ which was difficult to compute. 
 \item Using the approximation in $(4)$, compute the solution to \eqref{appro-reli} with an accuracy $\epsilon_1$ and a confidence $\beta_1$ as
  \begin{equation}\label{approx-reliability}
\min \left\{\gamma\;\Bigg|\; \frac{1}{N_1} \sum_{i=1}^{N_1} \mathbb{I}\left( f(\delta^i) \leq \gamma \right) \geq 1-\epsilon\right\}
 \end{equation}
\end{enumerate}
\begin{proof}
Let $f(\delta) \triangleq  ||\tilde{\Sigma}(\delta)||_{\ell_{2e},y_p \leftarrow y_r}^2$. If we substitute this definition in \eqref{appro-reli}, we obtain \eqref{uncertain:system} which solves for $\gamma_{OA}$. We use scenario based PEA to approximately compute $\gamma_{OA}$. To this end, let us define $\beta_1 \in (0,1)$ and $\epsilon_1 \in (0,1)$. Using these definitions, from step $(5)$ of PEA, we obtain $\gamma_{OA} $ with an accuracy $\epsilon_1$ and confidence $\beta_1$ by solving
\begin{equation*}
\min_{\gamma} \left\{ \gamma \Bigg| \frac{1}{N_1} \sum_{i=1}^{N_1} \mathbb{I}\left( ||\tilde{\Sigma}(\delta_i)||_{\ell_{2e},y_p \leftarrow y_r}^2 \leq \gamma \right) \geq 1-\beta\right\}
\end{equation*}
where $N_1$ is given by \eqref{sam1}.  Using the result of \textit{Lemma \ref{lemm1}}, $||\tilde{\Sigma}(\delta_i)||_{\ell_{2e},y_p \leftarrow y_r}^2$ can be obtained by solving the convex SDP \eqref{s2}. This concludes the proof.
\end{proof}
\section*{Proof of \textit{Lemma \ref{bound_OA_multi}}}\label{App3}
\begin{proof}
Necessity: Multiply both sides of the constraint of \eqref{problem:decoupled} by $N_1$. Then, for a bounded $\gamma_{OA}$ to exist, we require $\sum_{i=1}^{N_1} \mathbb{I}\left( \gamma_i \leq \gamma_{OA} , \; \forall a_i \in \ell_{2e}\right) \geq N_1-N_1\beta,$ should hold with a bounded $\gamma_{OA}$. This is satisfied only if $\lceil N_1(1-\beta)\rceil$ values of the set $\{\gamma_i\}_{i=1,\dots,N_1}$ are bounded.\\
Sufficiency: Assume that the values of the set $\{\gamma_i\}_{i=1,\dots,N_1}$, for the first $\left \lceil N_1(1-\beta)\right \rceil$ realizations is bounded and is unbounded for the other realizations. Let us choose $\gamma_{OA} = \max_{i=1,\dots,\lceil N_1(1-\beta) \rceil} \gamma_i$. Substituting the definition of $\gamma_i, \gamma_{OA}$ in the constraint of \eqref{problem:decoupled} yields
\begin{align*}
&\frac{1}{N_1}\left\{\sum_{i=1}^{\lceil N_1(1-\beta)\rceil} \mathbb{I}\left( \gamma_i \leq \gamma_{OA} \right) + \sum_{i=\lceil N_1(1-\beta)\rceil +1}^{N_1} \mathbb{I}\left( \gamma_i \leq \gamma_{OA} \right) \right\} \\
&=\frac{1}{N_1}\left\{ \lceil N_1(1-\beta) \rceil + 0 \right\} \geq 1-\beta.
\end{align*}
Thus, there exists a bounded $\gamma_{OA}$ which satisfies the constraint of \eqref{problem:decoupled}. This concludes the proof.
\end{proof}

\bibliographystyle{ieeetr}
\bibliography{ifacconf}

\clearpage
\end{document}